	\documentclass{amsart}
 	\usepackage{url}
	\usepackage{amsthm}
	\usepackage{graphicx}
	\newtheorem{theorem}{Lemma}
\begin{document}
\author[B.~Mielnik]{Bogdan Mielnik}
\address[B.~Mielnik]{Department of Physics, CINVESTAV--IPN\\ AP 14--740, Mexico City.}
\email[B.~Mielnik]{bogdan@fis.cinvestav.mx}
\author[J.~Fuentes]{Jes\'us Fuentes}
\address[J.~Fuentes]{Department of Physics, CINVESTAV--IPN\\ AP 14--740, Mexico City.}
\email[J.~Fuentes]{jfuentes@fis.cinvestav.mx}
\title{Squeezed Fourier Meets Toeplitz Algebras}
%
%
%
%
\begin{abstract}
We look for new steps on the dynamical operations that may squeeze simultaneously some families of quantum states, independently of their initial shape, induced  by softly acting external fields which might produce the squeezing of the canonical observables $q,p$ of charged particles. Also, we present some exactly solvable cases of the problem which appear in the symmetric evolution intervals permitting to find explicitly the time dependence of the external fields needed to generate the required evolution operators. Curiously, our results are interrelated with a simple, non--trivial, anti--commuting algebra of Toeplitz which describes the problem more easily than the frequently used Ermakov--Milne invariants.
\end{abstract}
\maketitle
%
%
%
%
\section{Introduction}
\label{s:intro}

The simplest phenomena of {\em squeezing} can occur in the evolution  of canonical variables $q,p$ for the non--relativistic time dependent, quadratic Hamiltonians in one dimension with variable elastic forces in either classical or quantum theory, namely
\begin{equation}
\label{eq:1}
H(\tau) = \frac{p^2}{2}+\beta(\tau)\frac{q^2}{2}
\end{equation}
where $q, p$ are the dimensionless canonical position and momentum,  $\tau$  is a dimensionless time, and we adopt units in which the mass $m = 1$; in quantum case also  $\hbar = 1$, $[q,p] = \mathrm{i}$. Our question is, whether the evolution generated by the Hamiltonian \eqref{eq:1} can at some moment produce a unitary operator $U$ transforming $q\rightarrow \lambda q,\ p\rightarrow \frac{1}{\lambda} p$ with $0\neq\lambda\in\mathbb{R}$, {\em i.e.} squeezing $q$ and expanding $p$ or {\it vice versa}? A more general question is, whether for any pair of quantum observables $a=a^\dagger$ and  $b=b^\dagger$, commuting to a number $[a,b]=\mathrm{i}\alpha$ $(\alpha\in\mathbb{R})$ an evolution operator can  transform $ a\rightarrow \lambda a,\  b\rightarrow \frac{1}{\lambda} b$, {\em i.e.} expanding $a$ at the cost of $b$. 


\subsection{Structure of the present work}

This manuscript is planned as follows. In the first part of section \ref{s:mathieu} we present and classify the analogues of optical operations for massive particles. Then, we introduce the generally unnoticed squeezing effects in the experimental conditions traditional for Paul's traps. Section \ref{s:fourier} is dedicated to the discussion of our principal idea on the squeezing by means of distorted Fourier transformations caused by sharp pulses of the oscillator potentials. Even if our solutions do not yet provide the exact laboratory prescriptions, they indicate that the squeezing effects can be experimentally approximated. This is thanks to the exact solutions described in sections \ref{s:exact}, where we report an elementary case of the Toeplitz algebra which permits to design the exact, soft equivalents of the desired dynamical effects. Then, in section \ref{s:test}, we study some elementary cases to test our method. Finally, in section \ref{s:finale} we report some fundamental hopes but also difficulties.

For convenience, our mathematical calculations are carried in dimensionless variables but the results are then translated into the physical units. As a companion for this paper, the reader is welcome to freely access to our codes and simulations \cite{algorithms} and to our previous discussion on this algorithm \cite{squeezing}.

%
%
%
%
\section{General Aspects}
\label{s:mathieu}


%

\subsection{The classical--quantum duality}

Below, we shall notice also the evolution laws induced by slightly more general Hamiltonians: 

\begin{equation}
\label{eq:2}
H(\tau) = \gamma(\tau) \frac{p^2}{2}+\beta(\tau)\frac{q^2}{2}
\end{equation}

Including the amply discussed Batman Hamiltonian.The evolution equations generated by \eqref{eq:2} in both classical and quantum cases imply exactly the same linear equations for either classical or quantum canonical variables: $\frac{dq}{d\tau} = \gamma(\tau) p(\tau)$, $\frac{dp}{d\tau} =-\beta(\tau)q(\tau)$, leading in any time interval $[\tau_0, \tau]$ to the identical  transformation of either classical or quantum canonical  pair, expressed by the same family of $2\times2$ symplectic {\em evolution matrices}  $u(\tau,\tau_0)$:

\begin{equation}
\label{eq:3}
\begin{Vmatrix} 
q(\tau) \\ p(\tau)
\end{Vmatrix} = u(\tau,\tau_0) 
\begin{Vmatrix} q(\tau_0)\\p(\tau_0)\end{Vmatrix}, \quad u(\tau_0,\tau_0) = 1,
\end{equation}

determined by the matrix equations
\begin{equation}
\label{eq:4}
\partial_{\tau} u(\tau,\tau_0) = \Lambda(\tau)u(\tau,\tau_0); \quad \Lambda(\tau) = \begin{Vmatrix}0&&\gamma(\tau)\\-\beta(\tau)&&0\end{Vmatrix}.
\end{equation}

The reciprocity between the classical and quantum pictures does not end up here.
It turns out that, in absence of spin, each unitary evolution operator $U(\tau,\tau_0)$ in $L^2(\mathbb{R})$ generated by the time dependent, quadratic Hamiltonian \eqref{eq:1} is determined, up to a phase factor, by the canonical transformation that it induces (see \cite{reed75, mielnik77, mielnik11, mielnik13}).



The behaviour of non--relativistic particles in one dimension subject to variable oscillatory fields \eqref{eq:1} was studied with the aim to describe the particle motion in Paul's traps \cite{paul90}, then in ample contributions of Glauber \cite{glauber92} and others. Some results about the  {\em squeezing} caused by variable electromagnetic fields were studied by Baseia {\em et al.} \cite{baseia92,baseia93}.   

Can one still achieve something more? As it seems, in some circumstances, instead of using the Wronskian to compare the independent solutions of \eqref{eq:1} an easier method might be to use the {\em classical--quantum duality} permitting to deduce the quantum evolution operators from the classical motion. The method could not work  for the general quantum motion, but it does for all quadratic Hamiltonians including \eqref{eq:1} -- cf. the encyclopaedic report by Dodonov \cite{dodonov02}. Here, we consider the variable external fields as the only credible source of such phenomenon. So, we skip all formal results obtained for time dependent masses, material constants, etc.

In quantum optics of {\em coherent photon states}, an important role belongs to the {\em parametric amplification} of Mollow and Glauber \cite{mollow67}. Yet, in the description of massive particles the Heisenberg's evolution of the canonical observables ({\em i.e.}, the trajectory picture) receives less attention, even though it allows to extend the optical concepts \cite{wolf04,wolf12}. This can be of special interest for charged particles in the ion traps, driven by the time dependent fields, coinciding or not with the formula of Paul \cite{paul90}. The most interesting here is the case of quite arbitrary periodic potentials.

\subsection{Elementary models}

Some traditional models illustrate  the above {\em duality doctrine}. Two of them seem of special interest.\bigskip

{\bf i }
The evolution of charged particles in the hyperbolically shaped ion traps\cite{paul90}.The Paul's potentials $\Phi({\bf x},t)$ in the trap interior generated by the voltage $\Phi(t)$ on the surfaces are either
\begin{equation*}
\Phi=\frac{e\Phi(t)}{r_0^2}\left(\frac{x^2}{2}+\frac{y^2}{2}-z^2\right) \quad \text{or} \quad \Phi=\frac{e\Phi(t)}{r_0^2}\left(\frac{x^2}{2}-\frac{y^2}{2}\right).
\end{equation*}
The problem then splits into the partial Hamiltonians of the type
\begin{equation*}
H(t)=\frac{p^2}{2}+\frac{e\Phi(t)}{r_0^2}\frac{q^2}{2},
\end{equation*}
where $q,p$ represent just one of independent pairs of canonical observables. 

By introducing the new dimensionless time variable $\tau = \frac{t}{T}$, where $T$ stands for an arbitrarily chosen time scale, each Hamiltonian is reduced to a particular case of \eqref{eq:1}:
\begin{equation}
\label{eq:4}
\tilde{H}(\tau) = H(t)T = \frac{\tilde{p}^2}{2}+\beta(\tau)\frac{\tilde{q}^2}{2}, \quad \beta(\tau) = \frac{e\Phi(t)T^2}{r_0^2 m}, 
\end{equation}
where $\beta(\tau)$ is dimensionless and the new  canonical variables $\tilde{q} =  \sqrt{\frac{m}{T}}q$ and $\tilde{p} = \sqrt{\frac{T}{m}} p$ are then expressed in the same units (square roots of the action), leading to the dimensionless evolution matrices $u(\tau,\tau_0)$ identical for the classical and quantum dynamics. So, {\em without even knowing} about the existence of quantum mechanics, the dimensionless quantities can be now constructed:
\begin{equation}
\label{eq:5}
q_d=\frac{\tilde{q}}{\sqrt{\hbar}} = q\sqrt{\frac{m}{\hbar T}}, \quad p_d = \frac{\tilde{p}}{\sqrt{\hbar}} =p\sqrt{\frac{T}{\hbar m}}
\end{equation}
and $H_d(t) = \frac{H(t)T}{\hbar}$, where $\hbar$ is an arbitrarily chosen action unit ({\em cf.} \cite{delgado98}) Indeed, for the quadratic Hamiltonians all results of {\em Mr. Tompkins in wonderland} by George Gamov, can be deduced just by rescaling time, canonical variables and the external fields. By knowing already about the quantum background of the theory, an obvious (though not obligatory) option is to choose $\hbar$ as the Planck constant (though the other constants  proportional to $\hbar$ of Planck are neither excluded\. By dropping the unnecessary indexes, one ends up with the evolution problem \eqref{eq:1}, with an arbitrary time dependent $\beta(\tau)$.
\bigskip

{\bf ii }
The similar dynamical law applies to charged particles moving in a time dependent magnetic field, given (in the first step of Einstein--Infeld--Hoffmann approximation \cite{infeld60}) by ${\bf B}(t ) = {\bf n} B(t)$, where ${\bf n}$ is a constant unit vector defining the central $z-$axis of a cylindrical solenoid. Since ${\bf B}(t)$ has the vector potential ${\bf A}(x,t) = \frac {1}{ 2} {\bf B}(t)\times{\bf x}$, the motion of the non--relativistic charged, {\em spinless} particle on the plane $\mathcal{P}_\perp$ perpendicular to ${\bf n}$, obeys the simplified Hamiltonian
\begin{equation*}
H(t) = \frac{1}{2m}\left[{\bf p}^2+\left(\frac{eB(t)}{2c}\right)^2{\bf x}^2\right],
\end{equation*}
where ${\bf p}$ and ${\bf x}$ are the pairs of canonical momenta and positions on $\mathcal{P}_\perp$. This, after using the dimensionless variables $\tau = \frac{t}{T}$, with $x,p_x$ and $y,p_y$ replacing $q,p$ in \eqref{eq:5}, leads again to a pair of motions of type \eqref{eq:1}, with the dimensionless $\tau, q , p$ and
\begin{equation}
\label{eq:8}
\beta(\tau)=\kappa^2(\tau)=\left(\frac{eTB(T\tau)}{2mc}\right)^2.
\end{equation}
In case $B(t)$ oscillates periodically with frequency $\omega$ the natural dimensionless time $\tau = \omega t$ leads again to a dimensionless $H_d = \frac{H(t)}{\hbar\omega}$, although the stability thresholds no longer obey the Strutt diagram \cite{bender78,mielnik10}.

Having say that, below, we shall be specially interested in the auxiliary operations of amplification or squeezing as a chance to apply the ideas of the demolition free measurements of Thorne {\em et al.} \cite{nondemolishing78,nondemolishing80}. Of course, not all techniques available awake an absolute confidence, including the decoherence \cite{bohm66}, the instability \cite{haroche98}, and the so called {\em delayed choice} \cite{wheeler84}. Recently, even the entangled states and their radiation effects are under some critical attention \cite{Lloyd12,ma,dolev,price}.


For convenience, our mathematical calculations are carried in dimensionless variables but the results are then translated into the physical units.
%
%
%
%
\section{Numerical Aspects}
\label{s:mathieu}
In quantum optics of {\em coherent photon states}, an important role belongs to the {\em parametric amplification} of Mollow and Glauber \cite{mollow67}. Yet, in the description of massive particles the Heisenberg's evolution of the canonical observables ({\em i.e.}, the trajectory picture) receives less attention, even though it allows to extend the optical concepts \cite{wolf04,wolf12}. This can be of special interest for charged particles in the ion traps, driven by the time dependent fields, coinciding or not with the formula of Paul \cite{paul90}. The most interesting here is the case of quite arbitrary periodic potentials.

\subsection{Classification}

For all periodic fields, $\beta(\tau +\footnotesize{\text{T}}) = \beta(\tau)$, the most important matrices \eqref{eq:2} are $u(\footnotesize{\text{T}} +\tau_0, \tau_0)$ describing the repeated evolution incidents. Since they are symplectic, their algebraic structure is fully defined just by one number $\Gamma=\text{Tr}~u(\footnotesize{\text{T}} +\tau_0, \tau_0) $, (without referring to the Ermakov--Milne invariants \cite{ermakov08,milne30,mayo02}). Though the matrices $u(\footnotesize{\text{T}} +\tau_0, \tau_0)$ depend on $\tau_0$, $\Gamma$ does not, permitting to classify the evolution processes generated by $\beta$ in any periodicity interval. The distinction between the three types of behaviour is quite elementary:

\begin{enumerate}
\item[{\bf I}]  If $\vert\Gamma\vert < 2$ the repeated $\beta$--periods, no matter the details, produce  an evolution matrix with a pair of eigenvalues $e^{\mathrm{i}\sigma}$ and  $e^{-\mathrm{i}\sigma}$ ($\sigma \in \mathbb{R}, 0<\sigma<\frac{\pi}{2}$) generating a stable (oscillating) evolution process. It allows the construction of 
the global creation and annihilation operators $a^+, a^-$ defined by the row eigenvectors of  $u(\footnotesize{\text{T}} +\tau_0, \tau_0)$, but featuring the evolution in the whole periodicity interval (compare \cite{mielnik11,mielnik10}).

\item[{\bf II}] If $\vert\Gamma\vert = 2$ the process generated by $\beta$ belongs to the {\em stability threshold} with eigenvalues $\pm1$ permitting to approximate a family of interesting dynamical operations ({\em cf.} the discussions in \cite{mielnik11,mielnik13,mielnik10}).

\item[{\bf III}] If $\vert\Gamma\vert > 2$ then each one--period evolution matrix has now  a pair of real non-vanishing eigenvalues, $\lambda^+ = \frac{1}{\lambda^-}$ with $\lambda^+ = e^\sigma$ and  $\lambda^-= e^{-\sigma}$ producing the squeezing of the corresponding pair of canonical observables $a^\pm$  defined again by the eigenvectors of  $u(\footnotesize{\text{T}} +\tau_0, \tau_0)$, that is, $a^+$ expands at the cost of contracting $a^-$ or vice versa.
\end{enumerate}

\subsection{The Mathieu squeezing}

The above global data seem more relevant than the description in terms of the {\em instantaneous} creation and annihilation operators which do not make obvious the  stability/squeezing thresholds. In the particular case of Paul's potentials with  $\beta(\tau ) = \beta_0 + 2\beta_1 \cos \tau$ the map of the squeezing boundary is determined by the Strutt diagram \cite{bender78}, traditionally limited to describe the ion trapping (in stability areas). Out of them are precisely the {\em squeezing effects} as stated in {\bf III}. 

To illustrate all of this, it is interesting to integrate \eqref{eq:3} for particular case of Paul's potential for $(\beta_0,\beta_1)$ out of the stability domain. Nevertheless, the squeezing cannot occur if $\beta(\tau)$ is symmetric in the operation interval \cite{wolf04, wolf12}, hence, we chose to integrate numerically \eqref{eq:3} for Paul's $\beta=\beta_0+2\beta_1\cos\tau$ in $\left[\frac{\pi}{2},\frac{5\pi}{2}\right]$ and $\beta_0,\beta_1$ varying in the second squeezing area of the Strutt diagram -- {\em cf.} \cite{bender78,mielnik10}. 

Then, we performed the scanning to localize the evolution matrices 
$u$ yielding the position squeezing. The results shown on figure \ref{f:strutt} generalize the numerical data of Ramirez \cite{mielnik10}. The continuous line on figure \ref{f:strutt} represents the $(\beta_0 , \beta_1)$ values for which the evolution matrix $u = u\left(\frac{5\pi}{2}, \frac{\pi}{2}\right)$ has the matrix element $u_{12} = 0$, while the dotted line contains the $(\beta_0,\beta_1)$ where $u_{21} = 0$.

The numbers $\lambda = u_{11}$, define the {\em squeezing} of $q$ completing the data obtained in \cite{mielnik10}. The particular matrices $u_1,u_2,u_3$ and $u_s$ obtained for the collection of pairs
\begin{equation*}
(\beta_0,\beta_1) = \{(1.054,0.646),(1.577,1.231),(1.774,1.454),(1.217,0.844)\},
\end{equation*}
are fully reported in \eqref{eq:9}. The points on the negative parts of the squeezing trajectory (continuous line), represent the inverted {\em squeezing} effects -- {\em cf.} the reinterpreted Strutt map \cite{mielnik10}. 

\begin{figure}[ht]
\begin{center}
\includegraphics[width=7cm]{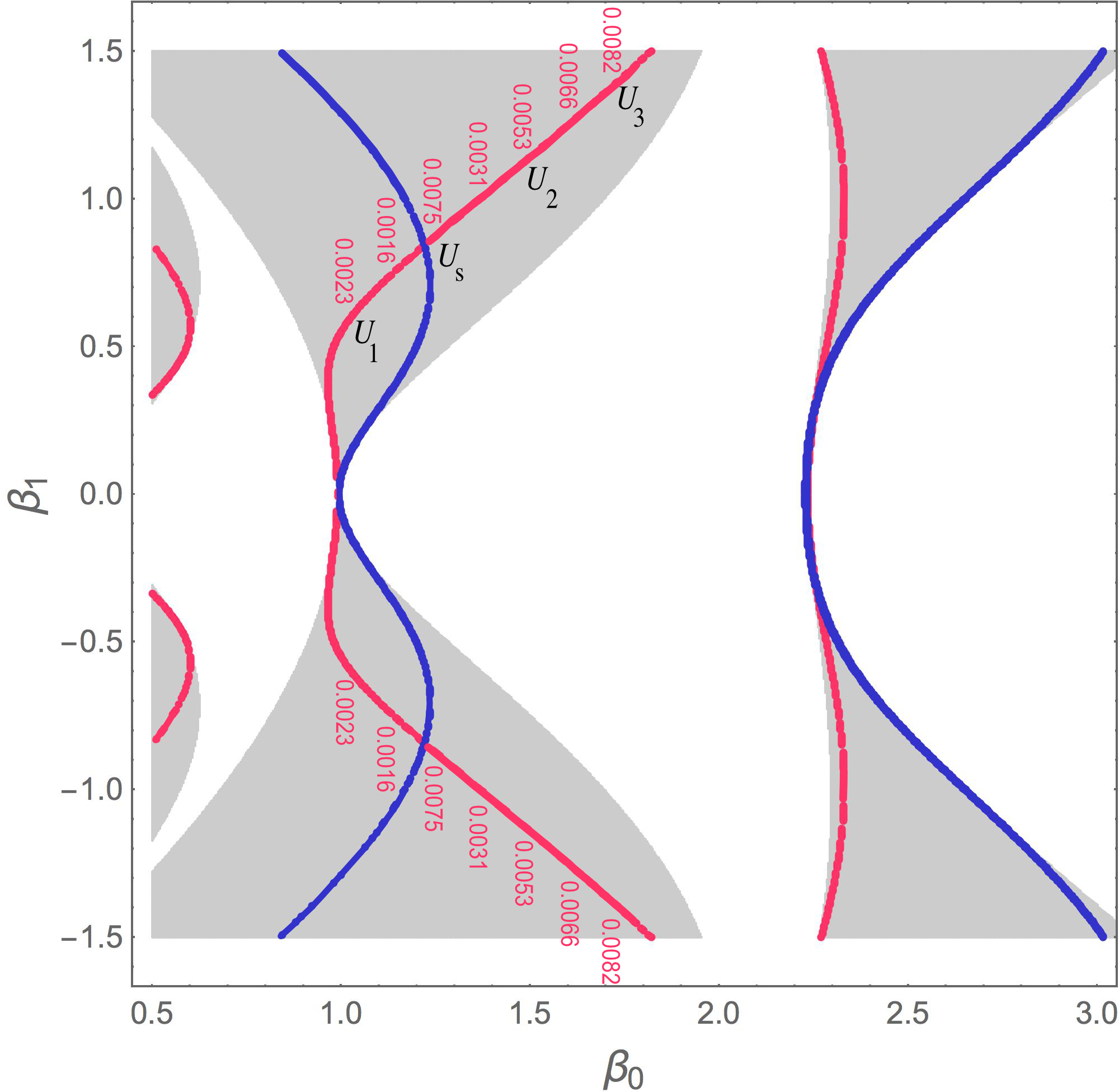}
\caption{The squeezing in Paul's trap. The continuous line collects the evolution matrices \eqref{eq:2} with $u_{12} = 0$ in the second squeezing area on the Strutt map.}
\label{f:strutt}
\end{center}
\end{figure}

Here, the additional data above the continuous line report slightly different effects when $q$ is squeezed (or amplified) at the cost of distinct canonical variables $a^- = u_{21}q + \frac{1}{\lambda} p$. As an example, we picked up the four evolution matrices representing various cases of squeezing granted by four pairs $(\beta_0,\beta_1)$ on figure \ref{f:strutt}. 

\begin{equation}
\label{eq:9}
\begin{split}
u_1  =\begin{Vmatrix}0.362 & 0.002\\-1.114 & 2.751\end{Vmatrix},\quad 
u_2 &=\begin{Vmatrix}0.175 & 0.005\\ 3.501 & 5.798\end{Vmatrix},\quad
u_3  =\begin{Vmatrix}0.216 & 0.008\\ 5.444 & 4.833\end{Vmatrix},\\ 
u_s &=\begin{Vmatrix}0.227 & 0.007\\ 0.000 & 4.394\end{Vmatrix}.
\end{split}
\end{equation}

While the matrix $u_s$ at the intersection of both curves in the upper (positive) part of the diagram represents the {\em squeezing coordinate} $q \rightarrow \lambda q$, $p\rightarrow \frac{1}{\lambda}p$ with  $\lambda \approx 0.227$, the corresponding intersection on the lower (negative) part represents an inverse operation with  $\frac{1}{\lambda} \approx 4.394$, {\em i.e.} amplifying $q$ and squeezing $p$. Henceforth, if the corresponding pulses were successively applied to two pairs of electrodes in a cylindric Paul's trap, then the particle state would suffer the sequence expansions of its $x$ variable with the simultaneous squeezing of $y$ and inversely, amplifying $y$ but squeezing  $x$. An open question is whether some new techniques of squeezing could appear by generalizing the operational techniques of high frequency pulses described in \cite{Itin,Martinez}.

Although analytic expressions in \cite{frenkel01} could yield to more exact results, our computer experiment, in fact, indicates that the phenomena of  $q,p$--squeezing can happen in the Paul's traps. Yet, they concern only the extremely clean Paul oscillations, without any laser cooling (crucial for the experimental  trapping techniques), nor any dissipative perturbations.    
Moreover,  the squeezing effects described by matrices \eqref{eq:9} are volatile, materializing itself only in sharply defined time moments, which makes difficult the observation of the phenomenon in the oscillating trap fields.

\subsection{Dimensions in Brief}

The boring problem of physical units brings, however, some additional data. For the dimensionless $\tau = \omega t$ the parameter  $\footnotesize{T}$ in \eqref{eq:4} 
is the period of the oscillating Paul's voltage on the trap wall  $\Phi(t) = \Phi_0 + \Phi_1 \cos \omega t \Rightarrow \beta(\tau) = \beta_0+2\beta_1\cos\tau$. Hence, the same dimensionless matrix $u_s$ of \eqref{eq:9} can be generated  in $\left[\frac{\pi}{2},\frac{5\pi}{2}\right]$, by physical parameters such that:
\begin{equation}
\label{eq:10}
\frac{e\Phi_0}{\omega^2 r_0^2m}=\beta_0, \quad \frac{e\Phi_1}{\omega^2 r_0^2m}=2\beta_1.
\end{equation}
 
In case of particles with fixed mass and charge, what can vary are the potentials $\Phi$ and the physical time $T=\frac{2\pi}{\footnotesize{\omega}}$ of the operations corresponding to the dimensionless interval $\left[\frac{\pi}{2},\frac{5\pi}{2}\right]$. Hence, for any fixed $r_0$, the smaller $\omega$ (and the longer $T$)  the smaller voltages $\Phi_0$ and $\Phi_1$ are sufficient to assure the same result (but only if too weak fields do not permit the particle to escape or to collide with the trap surfaces). For a proton ($m = m_p \simeq \text{1.67} \times 10^{-24}$g) in an unusually ample ion trap of $r_0 = 10$cm and in a moderately oscillating Paul's field with $\omega$ corresponding to a 3km long radio wave, one would have $\omega^2 r_0^2 m_p \simeq 10^{-12} \text{g}\frac{\text{cm}^2}{s^2} = \text{1.67} \times 10^{-12}\text{g}\frac{\text{cm}^2}{\text{s}^2} \simeq \text{1.04233}$eV, leading to the voltage estimations: 
$\Phi_0 \simeq \text{1.0423}\beta_0 \text{V} \simeq \text{1.268}$V and 
$\Phi_1 \simeq \text{2.0846}\beta_1 \text{V} \simeq \text{1.759}$V. In a still wider trap with $r_0 = 100$cm or, alternatively, with $r_0 =10$cm but the frequency ten times higher, then the voltages needed on the walls should be already one hundred times higher.
%
%
%
%
\section{Squeezed Fourier}
\label{s:fourier}

One of the simplest ways to construct the quantum evolution  operations is to apply  sequences of external  $\delta$--pulses interrupting some continuous evolution process ({\em e.g.} the free evolution, the harmonic oscillation, etc. \cite{fernandez92,ammann97,viola99,viola03}). However,  the method is  obviously limited by the practical impossibility of applying the $\delta$--pulses of the external fields. In case of squeezing, a more regular method could be to compose some evolution incidents which belong to the equilibrium zone {\bf I} but their products do necessarily not. One of the chances is to use the fragments of time independent oscillator fields \eqref{eq:1} with the elastic forces $\beta = \kappa^2 = \text{constant}$, generating the symplectic rotations:
\begin{equation}
\label{sprot}
u=\begin{Vmatrix}\cos\kappa\tau & \frac{\sin\kappa\tau}{\kappa} \\ -\kappa \sin\kappa\tau & \cos\kappa\tau\end{Vmatrix}.
\end{equation}
Their simplest cases obtained for $\cos\kappa\tau=0$ are the {\em squeezed Fourier transformations}
\begin{equation} 
\label{sqF}
u=\begin{Vmatrix}0&\pm\frac{1}{\kappa}\\ \mp \kappa&0\end{Vmatrix}.
\end{equation}
Following the proposal of Fan and Zaidi \cite{fan88}, and Gr\"ubl \cite{grubl89} it is enough to apply two such steps with different $\kappa$-values to generate the evolution matrix
\begin{equation}
\label{squeeze}
u_\lambda = \begin{Vmatrix}0&\pm\frac{1}{\kappa_1}\\ \mp\kappa_1&0\end{Vmatrix}\begin{Vmatrix}0&\pm\frac{1}{\kappa_2}\\ \mp\kappa_2&0\end{Vmatrix}=\begin{Vmatrix}\lambda&0\\ 0&\frac{1}{\lambda}\end{Vmatrix}; \quad \lambda = -\frac{\kappa_2}{\kappa_1}
\end{equation}
which produces the squeezing of the canonical pair: $q \rightarrow \lambda q\ ,p \rightarrow\frac{1}{\lambda}p$, with the effective evolution operator: 
\begin{equation*}
U_\lambda = \exp\left[-\mathrm{i}\sigma \frac{pq+qp}{2}\right], \quad \sigma = \ln\lambda. 
\end{equation*}
It requires, though, two different  $\kappa_1 \neq\kappa_2$ in two different time intervals divided by a sudden potential jump. (Here, the times $\tau_1$ and $\tau_2$ can fulfil {\em e.g.} $\kappa_1\tau_1 = \kappa_2\tau_2 = \frac{\pi}{2}$ to assure that both $\kappa_1$ and $\kappa_2$ grant two distinct squeezed Fourier operations in their time intervals.) If one wants to apply two potential steps on the null background, it means at least three jumps ($0 \rightarrow \kappa_1 \rightarrow \kappa_2 \rightarrow 0$). How exactly can one approximate a jump of the elastic potential? Moreover, each $\kappa$-jump implies an energy transfer to the micro-particle \cite{grubl89}. So, could the pair of generalized Fourier operations in \eqref{squeeze} be superposed in a {\em soft way} with an identical end result? In fact, the recent progress in the inverse evolution problem shows the existence of such effects.
%
%
%
%
\section{Devising Exact Solutions}
\label{s:exact}

Though the exact expressions \eqref{squeeze} were already known, it was not noticed that they can be generated by the simple anti--commuting {\em Toeplitz algebra} of $2\times2$ {\em equidiagonal symplectic matrices} $u$ with $u_{11} = u_{22} = \frac{1}{2} \text{Tr}~u$. It turns out that for any two such matrices $u, v$  their anti-commutator $uv+vu$ as well as  the symmetric products $uvu$ and $vuv$ belong to the same family. {\em Toeplitz matrices} have inspired a lot of research (see \cite{bottcher00, trefethen05, Deift} and the references therein)  though apparently, without paying  attention to their simplest quantum control sense. In our case, even without eliminating  jumps in \eqref{squeeze} they give an additional flexibility in constructing the squeezed Fourier operations as the symmetric products of many little symplectic contributions 
\eqref{sprot} with different $\beta$'s acting in different time intervals. Thus, by using the fragments of the symplectic rotations $v_k$ caused by the Hamiltonians \eqref{eq:1} with some fixed $\beta = \beta_i$ in time intervals $\Delta\tau_i$, with $i=0,1,2,\ldots$, one can define the symmetric product
\begin{equation}
\label{eq:12}
u=v_n\cdots v_1v_0v_1\cdots v_n,
\end{equation}
again, symplectic and {\em equidiagonal} ({\em i.e.} of the simplest Toeplitz class), with $u_{11} = u_{22} = \frac{1}{2}\text{Tr}~u$. Whenever \eqref{eq:12} achieves $\text{Tr}~u = 0$, the matrix $u$ becomes  {\em squeezed Fourier}. The continuous equivalents can be readily obtained. Indeed, it is
enough to assume that the amplitude $\beta(\tau)$ is symmetric around a certain  point $\tau = 0$, {\em i.e.}, $\beta(\tau) = \beta(-\tau)$. By considering the limits of little jumps $\mathrm{d}u$ caused by applying the contributions $\mathrm{d}v = \Lambda(\tau)\mathrm{d}\tau$  from the left and right sides, one arrives at the differential equation for $u = u(\tau,-\tau)$ in the expanding  interval $[-\tau, \tau]$, that is
\begin{equation}
\label{eq:13}
\frac{\text{d}u}{\text{d}\tau} = \Lambda(\tau) u+u\Lambda(\tau).
\end{equation} 
Certainly, the matrix $\Lambda(\tau)$ might read as expressed in \eqref{eq:3}, nevertheless, we can consider a slightly more general fashion of it for a better insight of our problem, {\em i.e.}
\begin{equation*}
\Lambda(\tau) = \begin{Vmatrix} 0 & \gamma(\tau) \\ -\beta(\tau) & 0\end{Vmatrix},
\end{equation*}
where $\gamma(\tau)$ is, in principle, an arbitrary function but chosen carefully so as to achieve the desired squeezing effect.

Now, because of the anti--commuting form of \eqref{eq:13}, we can obtain easily an exact solution to the equation
\begin{equation}
\label{eq:14}
\frac{\text{d}u}{\text{d}\tau} = (u_{21}\gamma - u_{12}\beta)\mathbb{I} + \text{Tr}~u\,\Lambda.
\end{equation}
In case of symmetric $\beta$, this determines explicitly the matrices $u = u(\tau,-\tau)$ for the expanding $[-\tau,\tau]$ in terms of just one function $\theta(\tau) \equiv u_{12}(\tau,-\tau)$. In fact, since \eqref{eq:14} implies the same differential equation for $u_{11}$ and $u_{22}$, {\em i.e.} $\frac{\text{d}}{\text{d}\tau}u_{11}=\frac{\text{d}}{\text{d}\tau} u_{22} =u_{21}\gamma-u_{12}\beta$, and since $u_{11} = u_{22} = 1$ at $\tau = 0$, then $u_{11} = u_{22} = \frac{1}{2}\gamma\text{Tr}~u = \frac{1}{2}\theta'$ for $u$ in any symmetric interval, $[-\tau,\tau]$. Moreover, since $u = u(\tau,-\tau)$ are symplectic, {\em i.e.} $\text{Det}~u = \left[\frac{1}{2}\theta'\right]^2- u_{21}\theta = 1$, one obtains
\begin{equation}
\label{u21}
u_{21}=\frac{\left[\frac{1}{2}\theta'\right]^2-1}{\theta}.
\end{equation}
Hence, \eqref{eq:14} defines the amplitude $\beta(\tau)$ which has to be applied to create the matrices $u=u(\tau,-\tau)$. Indeed 
\begin{equation}
\label{eq:15}
u_{12}\beta = u_{21}\gamma -\frac{d}{d\tau}u_{11}
\end{equation}
and since $u_{12}=\theta$, then $\frac{d}{d\tau}u_{11}=\frac{\theta''}{2}$. Thus $u_{21}$ is given by \eqref{u21} and consequently,
\begin{equation}
\label{beta} 
\beta = \gamma \frac{\left[\frac{1}{2}\theta'\right]^2-1}{\theta^2} -\frac{\theta''}{2\theta}.
\end{equation}

This solves  the  symmetric evolution problem for $u$ and $\beta$ in any interval $[-\tau,\tau]$ in terms of an almost arbitrary function $\theta(\tau)$, restricted by non--trivial conditions in single points only. Hence, \eqref{beta} is indeed an exact solution of the  inverse evolution problem, offering $\beta(\tau)$ in terms of the function $\theta(\tau) = u_{12}(\tau,-\tau)$  representing the evolution matrices for the expanding (or shrinking) evolution intervals $[\tau,-\tau]$. 
Note though that the dependence of $u(\tau, \tau_0)$ on $\beta(\tau)$ given by \eqref{beta} in any non--symmetric interval $[\tau, \tau_0]$  still requires an additional integration of \eqref{eq:3} between $\tau_0$ and $\tau$. Some simple algebraic relations of $\beta,\gamma$ and $\theta$ are worth attention.

\begin{theorem} Suppose $\beta(\tau)$ is given by \eqref{beta}, in a certain interval $[-\text{\footnotesize{T}},\text{\footnotesize{T}}]$, where $\theta(\tau)$ is continuous and at least three times differentiable. The conditions which assure the continuity, differentiability of $\beta$ and the dynamical 
relations between $\beta,\gamma$ and $\theta$ are then:
\begin{enumerate}
\item[1] At any point $\tau$ where $\theta(\tau) = 0$, there must be $\theta'(\tau) = \pm2$.
\item[2] If, moreover, $\theta'''(\tau) = 0$ then also $\beta'(\tau) = 0$.
\item[3] At any point $\tau$ where $\theta(\tau) \neq 0$ but $\theta'(\tau)=0$, the matrix  \eqref{eq:14} for $[-\tau,\tau]$ represents the squeezed Fourier transformation with $\beta(\tau)$ at the end points  given by  $-\beta(\tau)\theta^2=\frac{1}{2}\theta''\theta + \gamma(\tau)$.
\end{enumerate} 
\end{theorem}

\begin{proof}[Proof] It follows straightforwardly by applying \eqref{eq:15}. In particular, since \eqref{eq:14} and the initial condition grants $u_{11}=u_{22}=\frac{1}{2}\theta'(\tau)$ then, whenever $\theta'=0$, both $u_{11}=u_{22}=0$ implying $u_{12}=b\neq0, u_{21}=-\frac{1}{b}$; which is the general form of the {\em squeezed Fourier} transformation.  Simultaneously, \eqref{beta} simplifies and the value of $\beta(\tau)$ fulfils $-\beta(\tau)\theta^2=\frac{1}{2}\theta''\theta +\gamma(\tau) \Rightarrow b\beta(\tau) + \frac{\theta''}{2} +\frac{1}{b}\gamma(\tau)=0$. In particular, if  $b\theta''(\tau) = -2$, then $\beta(\tau) = 0$.
\end{proof}

Here, we have used the simplest non--trivial case of Toeplitz algebra. This solves the inverse evolution problem for $\beta(\tau)=\kappa^2(\tau)$ in terms of $\theta(\tau)$ and $\gamma(\tau)$, without any auxiliary invariants. However, its purely comparative sense should be highlighted. For a fixed pair of canonical variables $q,p$ it does not give the causally progressing process of the evolution, but rather compares the evolution incidents in a family of expanding  intervals $[-\tau, \tau]$.  Should one like to follow the causal development of the classical/quantum systems, the Ermakov--Milne equation \cite{ermakov08,milne30} might be useful. An interrelation between both methods waits still for an exact description. It is not excluded that the anti--commutator algebras can help also in some higher dimensional canonical problems.\footnote{It seems truly puzzling that this extremely simple case of anti--commuting Toeplitz algebra was never associated with the variable oscillator evolution.}
%
%
%
%
\section{Testing the Algorithm}
\label{s:test}

As already checked in \cite{mielnik13}, there exist polynomial models of $\theta(\tau)$  making possible the soft generation of the squeezed Fourier transformations. The polynomials of $\tau$, however, are just a formal exercise. As it seems, it would be more natural to apply the harmonically oscillating $\theta$ functions. As the most elementary case, let us consider the evolution guided by $\theta$ with only four frequencies. In dimensionless variables we have:
\begin{equation}
\label{eq:16}
\theta(\tau) = a_1 \sin \tau + a_3 \sin{3\tau} + a_5 \sin5\tau + a_7 \sin7\tau.
\end{equation}

Note that for $\theta(\tau)$ antisymmetric, the corresponding $\beta(\tau)$ defined by \eqref{eq:15} is symmetric around $\tau=0$. In order to generate the soft squeezed Fourier operations with $u_{12}= \pm \kappa = b$ at the ends of the symmetric interval $\left[-\frac{\pi}{2},\frac{\pi}{2}\right]$, the function $\theta$ given by \eqref{eq:16} must satisfy the conditions of Lemma 2, thus we obtain the following relations:
\begin{equation}
\label{eq:17}
\begin{split}
&\theta \left(\frac{\pi}{2}\right) = a_1 - a_3 + a_5 -a_7= b,\\
&\theta' (0) = a_1 +3 a_3 + 5a_5 +7a_7= 2,\\
&\theta'' \left(\frac{\pi}{2}\right) = -a_1 +9 a_3 - 25 a_5 = -\frac{2}{b}\gamma -2b\beta_0, \\
&\theta''' (0) = a_1 + 27 a_3 + 125a_5 + 343 a_7 = -c.
\end{split}
\end{equation}

The interrelation between the harmonic $\theta$ in \eqref{eq:16} and the corresponding physical $\beta$ given by \eqref{beta} is not completely trivial, but reduces to a purely algebraic problem, where the first identity grants the non-singularity of $\beta$ at $\tau=0$, the second one defines the magnitude $b$ of the {\em Fourier squeezing} depending on the whole trajectory, and $\beta_0$ and $\gamma$ define the symmetric values of the  amplitude $\beta(\tau)$ at $\pm\frac{\pi}{2}$. Equations \eqref{eq:17} are then fulfilled by:
\begin{equation}
\label{eq:18}
\begin{split}
a_{1} &= \frac{b [58 + c + 5 b (21 - 2 \beta_{0})] - 10 \gamma}{128 b}, \\
a_{3} &= \frac{b [74 + c -   b (35 + 2 \beta_{0})] - 2  \gamma}{128 b}, \\
a_{5} &= \frac{b [22 - c + 3 b (-7 + 6 \beta_{0})] + 18 \gamma}{384 b}, \\
a_{7} &=-\frac{b [26 + c + 3 b (-5 + 2 \beta_{0})] + 6  \gamma}{384 b}, 
\end{split}
\end{equation}
with $b,c\neq 0$ two suitable real constants. Note that our assumed (antisymmetric) $\theta(\tau)$ represents  $u_{12}(-\tau,\tau)$ for $\tau>0$, while the obtained (symmetric) $\beta(\tau)$ defines the field amplitude in the whole symmetry interval. 

\subsection{The simplest cases}

At this stage, we are able to sketch some examples of the amplitudes following the remarks previously discussed.  

For instance, in regard to the left panel of figure \ref{fig:2}, with vanishing initial value $\beta_0$, both solid and dashed curves (which do not cross to the negative values) are adequate to achieve the squeezed Fourier operations by time dependent magnetic fields with $B(t)\sim\sqrt{\beta(t)}$. These ones, if superposed softly in two consecutive intervals $[-\frac{\pi}{2},\frac{\pi}{2}]$ and $[\frac{\pi}{2},\frac{3\pi}{2}]$ can generate the $q,p$--squeezing effect (see \cite{squeezing}). However, as shown in figure \ref{fig:3}, we can also  achieve a squeezing effect by the application of a single amplitude along the whole interval.

In a similar fashion, for the right panel of figure \ref{fig:2}, both solid and dashed pulses grant the magnetic squeezed Fourier in their first action interval $[-\frac{\pi}{2},\frac{\pi}{2}]$. In the next $\tau$--interval $\left[\frac{\pi}{2}, \sqrt{ \frac{5}{2}} \pi \right]$ they all reduce themselves to the constant $\beta_0$ generating the same squeezed Fourier with $b_0=\frac{1}{10}$. The two of these fragments correspond to the stability zone {\bf I}, and together produce the amplification $\lambda\simeq -1.71$, {\em i.e.} if composed softly – refer to \cite{squeezing} for further examples.

The amplitudes which cross by zero are not suitable to generate the squeezed (or amplification) effect, rather, one can take advantage of these in case of ion traps. Another approach to evaluate those amplitudes suitable to generate {\em squeezing} is briefly discussed in the appendix \ref{s:appendix}.

\begin{figure}[h!]
\begin{center}
$\vcenter{\hbox{\includegraphics[width=6.25cm]{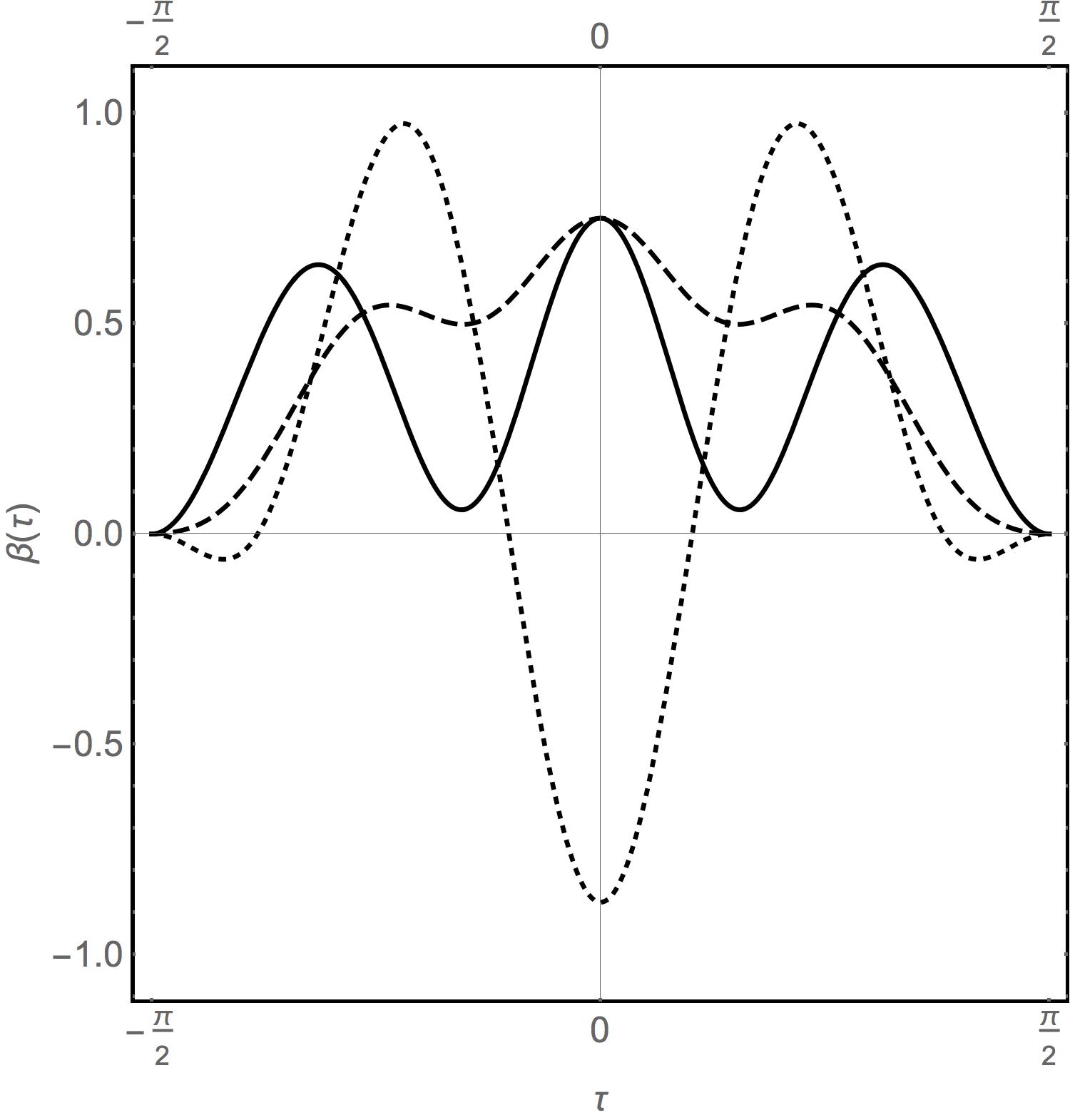}}}$
$\vcenter{\hbox{\includegraphics[width=6.25cm]{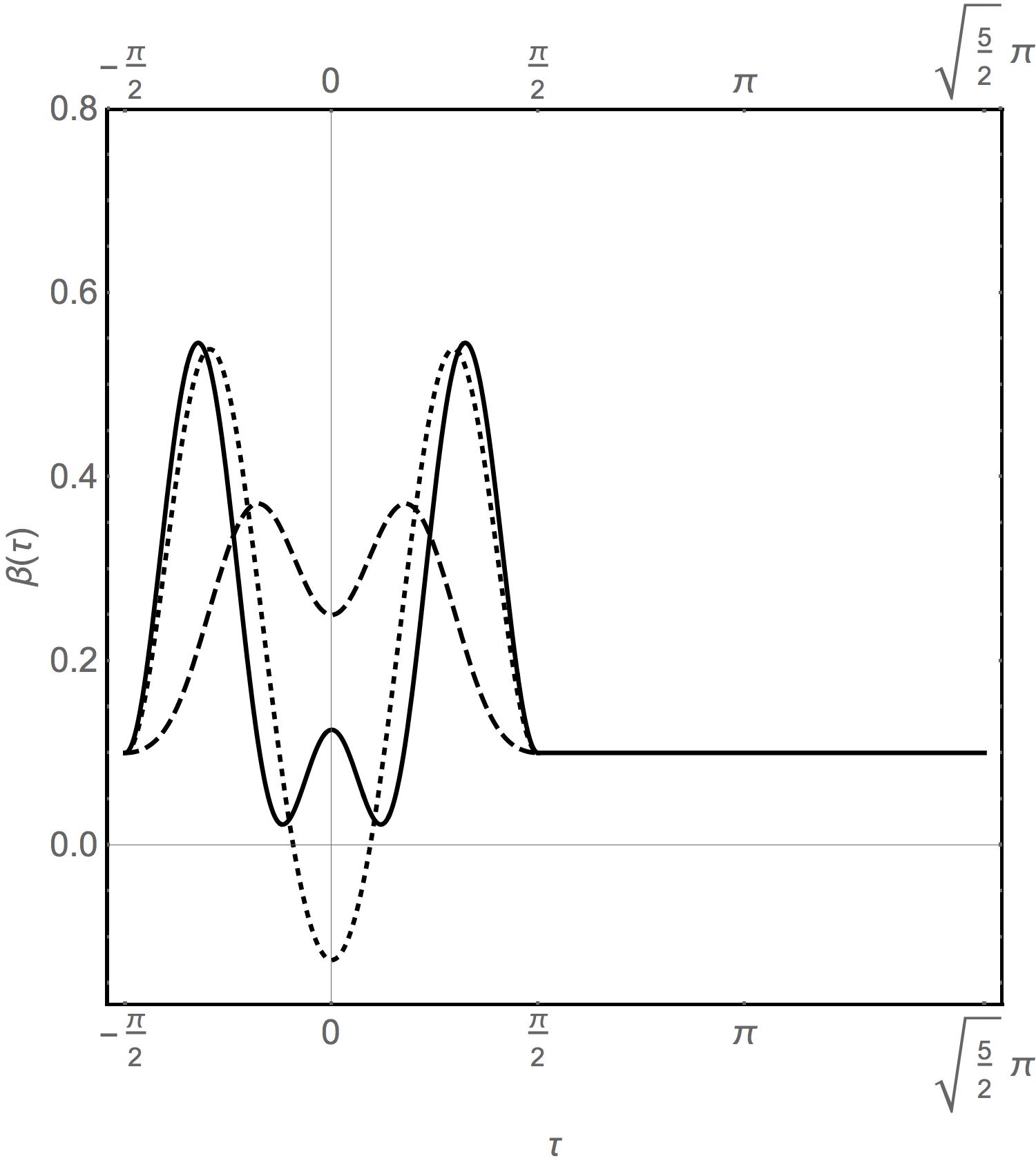}}}$
\\
$\vcenter{\hbox{\footnotesize (a) \hspace{5.5cm} (b)}}$
\caption{(a) The form of three symmetric $\beta$--amplitudes vanishing softly at the borders of the interval $[-\frac{\pi}{2},\frac{\pi}{2}]$ with $\beta_{0}=0$, $\gamma=\sin^{2}\tau$ and: $b=2, c=-3$ (solid),  $b=\frac{7}{4}, c=-3$ (dashed), $b=\frac{9}{5}, c=\frac{7}{2}$ (dotted). (b) The $\beta$--amplitudes satisfying \eqref{eq:17} with $\beta_0=\frac{1}{10}$, $\gamma$=1 and: $b=\frac{43}{20}, c=-1$  (solid), $b=\frac{37}{20}, c=-2$ (dashed) and $b=\frac{43}{20}, c=1$ (dotted).} 
\label{fig:2}
\end{center}
\end{figure}

We deliberately chose the case of partial $\beta$--amplitudes starting and ending up with $\beta(-\frac{\pi}{2})=\beta(\frac{\pi}{2})=\beta(\frac{3\pi}{2})=0$ to illustrate the flexibility of the method. In fact, we could notice that some programs of frictionless driving seem to exclude the  continuity  at the beginning and  at the end of the transport operation or even assume some sharp steps in the interior. Thus, in an interesting report by X. Chen {\em et al.} \cite{muga11} the authors present an operation modifying the harmonic oscillator $H_0$ by adding  some perturbation $H_1$, which vanishes before and after the operation, it can also appear or disappear suddenly -- likewise in \cite{muga10}. However, the interruption of an adiabatic process by a new potential  which can suddenly  {\em jump to existence} might be good to achieve the speed and efficiency of frictionless driving but  not the adiabatic qualities (see the results of Gr\"ubl \cite{grubl89}).

\subsection{Uncertainty shadows}

Certainly, the $\beta$--pulses determined by \eqref{eq:16}--\eqref{eq:18} do not yet define the actual trajectory inside the whole interval ({\em i.e.} for $\tau\neq -\tau_0,\tau $), which must be determined by a separate computer simulation – see \cite{algorithms}.

With this aim, we integrated the matrix equation \eqref{eq:3} for $u(\tau,\tau_0)$ subject to the initial condition $u(\tau_0,\tau_0)=\mathbb{I}$, where $\tau_0=-\frac{\pi}{2}$ starts the  evolution interval. We then continued the integration up to $\tau = \frac{35\pi}{32}$ obtaining a family of $2\times 2$ evolution matrices which draw a congruence of trajectories along the whole evolution interval, {\em cf.} figure \ref{fig:3}(a). As one can see, the trajectories indeed paint an image of the pair of squeezed Fourier in the whole interval and the coordinate squeezing at the very end. The final result is the $q$--amplification with $\lambda \approx -\text{1.14}$ corresponding to the solid $\beta$--amplitude on figure \ref{fig:2}. If generated by a magnetic field in a cylindrical solenoid, it would mean the $\lambda$--expansion of both coordinates $q=x,y$. Thus, if the operation is performed for an initial Gaussian packet, in a cylindric Paul's trap or solenoid of dimensionless radius large enough ({\em e.g.} $r_0>10$), this means that there exists a little probability that the particle collides with the trap or solenoid wall.

\begin{figure}[ht]
\begin{center}
$\vcenter{\hbox{\includegraphics[width=6.25cm]{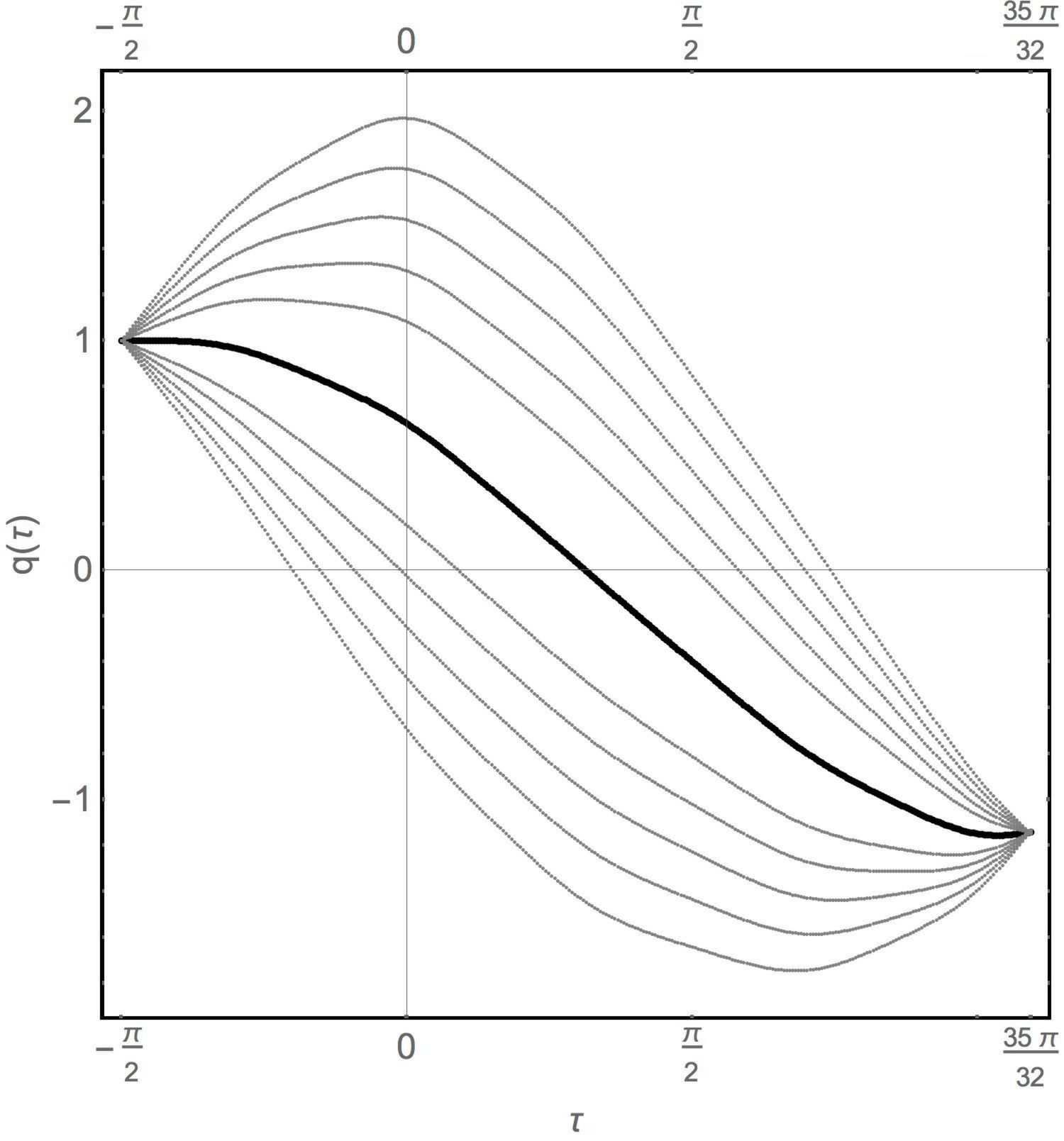}}}$
$\vcenter{\hbox{\includegraphics[width=6.25cm]{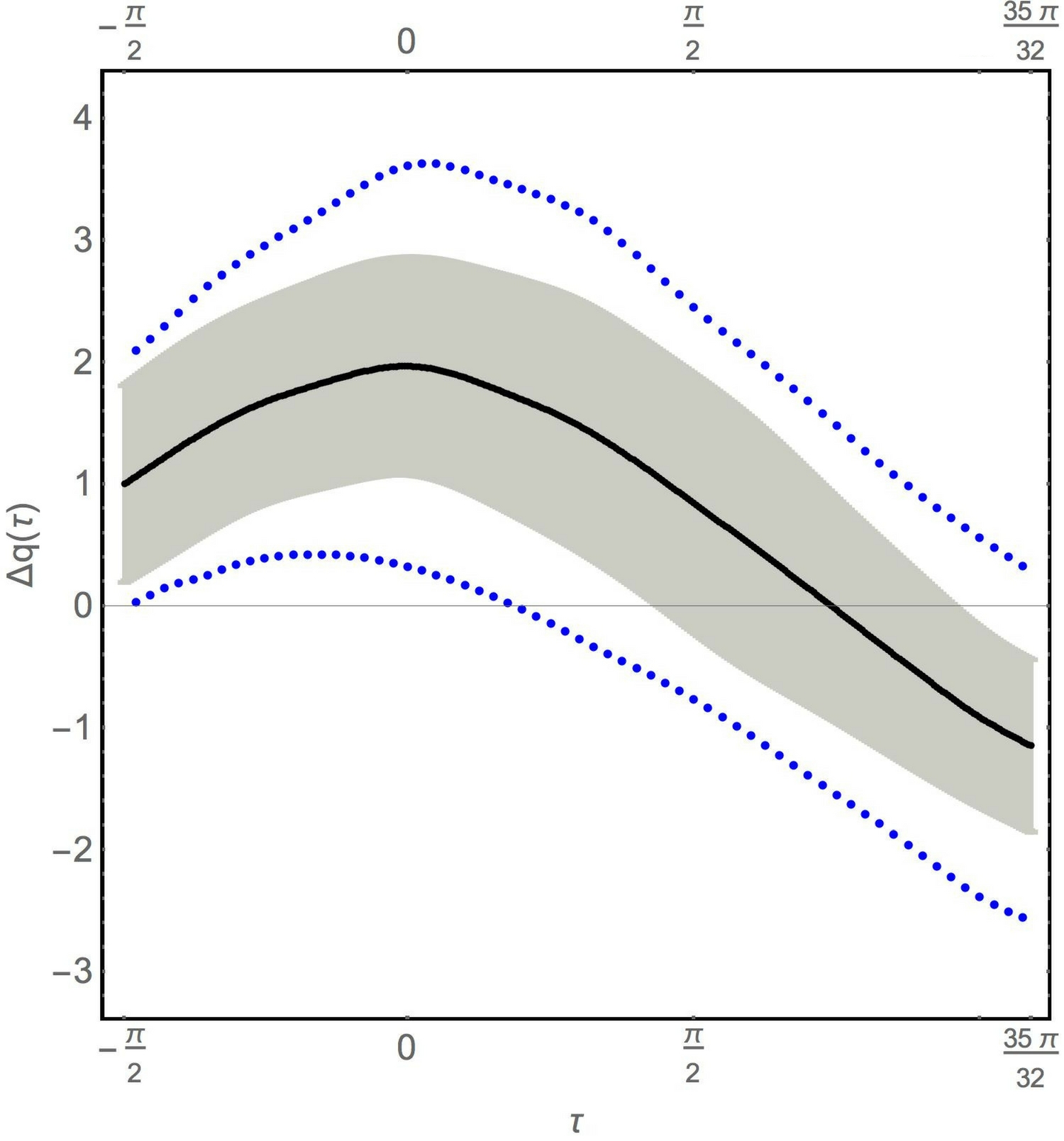}}}$
\\
$\vcenter{\hbox{\footnotesize (a) \hspace{5.5cm} (b)}}$
\caption{(a)  The trajectories generated by the single application of a squeezed Fourier operation induced by the solid $\beta$--amplitude of figure \ref{fig:2}(a)  along the asymmetric interval $[-\frac{\pi}{2},\frac{35}{32}\pi]$. (b) The uncertainty shadow $\Delta q$ for the upper trajectory sketched on part (a), determined for the numerically calculated $u_{11}$ and $u_{12}$ according to \eqref{eq:3} with the initial values $\Delta q =\Delta p = \frac{1}{2}$, marking the 0.999 threshold of packet probability.}
\label{fig:3}
\end{center}
\end{figure}

The above time dependent family  $u(\tau,\tau_0)$  in the interval  $\left[-\frac{\pi}{2},\frac{35\pi}{32}\right]$,  permits also to observe the going of the position and momentum uncertainties on the trajectory. As an example, we took one of the the most elementary Gaussian wave functions in $L^2(\mathbb{R})$ centred at $x = q_0$ with initial velocity $p_0$, that is
\begin{equation}
\label{eq:19}
\Psi(x,0) = A \exp\left[ i p_0 (x-q_0) \right]\exp\left[-\kappa \frac{(x-q_0)^2}{2} \right], \quad A=\left(\frac{\kappa}{\pi}\right)^{\frac{1}{4}}.
\end{equation}
For $\kappa=1= q_0=1$, and for varying $p_0$  the packet centre will draw exactly the family of trajectories on figure \ref{fig:3}(a) and the simple calculation with the initial uncertainties $(\Delta q)^2 = (\Delta p)^2 = \frac{1}{2}$, leads to:
\begin{equation}
\label{eq:22}
\vert \Delta q(\tau) \vert^2 = \frac{1}{2}\left[u_{11}^2(\tau) + u_{12}^2(\tau)\right].
\end{equation}

We then used the square root of \eqref{eq:22} to correct the upper trajectory of figure \ref{fig:3}(a) ($p_0=1$) by its {\em uncertainty shadow}, as shown on figure \ref{fig:3}(b). These results suggests that the main part of the evolving packet is contained within a wider dimensionless belt, {\em e.g.}  $\vert q (\tau)\vert < 10$ in the whole evolution interval $\left[-\frac{\pi}{2},\frac{35\pi}{32}\right]$. Characteristically, the uncertainty effects are most visible in the middle of the trajectory, for $\tau=\frac{\pi}{2}$ where two distinct squeezed Fourier meet, but they stick to the final amplified state at $\tau=\frac{35\pi}{32}$.


In fact, data on $\Delta q$ collected from figure \ref{fig:3} can additionally provide the more detailed statistical information. Our initial packet is not an eigenstate of any instantaneous Hamiltonian \eqref{eq:1}, but it is Gaussian and so are the transformed states $\Psi (x,\tau)$.

Thanks to the classical--quantum duality, the evolution matrix $u(\tau,\tau_0)$ determines also the evolved quantum state  $\Psi$. Some difficulty consists only in expressing $\Psi$ exclusively in $x$-representation. Yet, generalizing the already known results \cite{cohen}, we could obtain an explicit expression for the $x$-probability density $|\Psi(x,\tau)|^2$ at an arbitrary $\tau$, namely
\begin{equation}
\label{probability}
|\Psi(x,\tau)|^2=\frac{1}{\sqrt{\pi} \Delta q}\text{exp}{-\frac{(x-\langle q\rangle)^2}{|\Delta q(\tau)|^2}}=\frac{\sqrt{2}}{ \sqrt{\pi \left[u_{11}^2 + u_{12}^2\right]}}\text{exp}{-2\frac{(x-u_{11})^2}{u_{11}^2 + u_{12}^2}},
\end{equation}
which can be compared with the formula (17), complement G1 in \cite{cohen}, for the free packet propagation. Our hypothesis is, that our formula \eqref{probability}, not limited to the free packets, is the next step permitting to express the probabilities for the Gaussian states in terms of matrices $u_{kl}$  in all cases of time dependent elastic forces.

{\bf  Remark 1}. Our construction differs slightly from the other ones used to generate the squeezing by time dependent oscillator potentials. Up to now, the algorithms for the soft squeezing were designed for non-vanishing initial and final values $\beta=\beta_0>0$ and $\gamma=\sin^2\tau$. Our development somehow, permits to avoid the difficulty, since it allows $\beta = 0$ at the beginning and at the end.  However, the construction deduced from  \eqref{eq:16} allows to generate the same effects for arbitrary initial and final values of $\beta$. As an example, we quote below an analogous case for a non-vanishing pair of initial and final $\beta$ values.  

We therefore checked that the squeezing of the wave packets can be as  exactly induced by just modifying the orthodox harmonic potential to produce the squeezed Fourier effect in  the operation interval $[-\frac{\pi}{2},\frac{\pi}{2}]$, conserving the same $\beta=\beta_0>0$ at both ends. 



{\bf Remark 2}. A certain surprise are the extremely delicate values of the squeezing effects and the corresponding electric and magnetic fields -- see the orders of magnitude in \cite{squeezing}. Can so weak interactions keep the particle and dictate its unitary transformations? Without entering deeper into the  discussion, let us only notice that the extremely weak fields could be of importance even in our own existence \cite{penrose94,hagan02,frixione14}.

%
%
%
%
\section{The Fundamental Aspects in Little}
\label{s:finale}
In spite of imperfections, we feel attracted into a kind of {\em what if story}. The problem is, whether the existing difficulties to achieve the squeezing are purely technical or they mean some fundamental barrier. If no barrier exists, and the operations can indeed be performed (or at least approximated), the implications could be of some deeper interest.

If the squeezing of the wave packets in $L^2(\mathbb{R}^n)$ defined by $\psi(x) = (\sqrt{\lambda})^{-n} \psi\left(\frac{x}{\lambda}\right)$ (in the lowest dimensions $n = 1, 2, 3)$ with $\vert\lambda\vert <1$ could be achieved as a unitary evolution operation, it would imply that no fundamental limits exist to the possibility of shrinking the particle in an arbitrarily small interval (surface, or volume). 

Some authors believe that such localization must fail at extremely small scale below Planck distance. Certainly, it is difficult to dismiss {\em a priori} the doubts. However, in many fundamental discussions, the {\em Planck distance} is used as a magic spell which permits one to formulate almost any hypothesis free of consistent conditions. Yet, if one truly believes that quantum mechanics is a linear theory, then even the most concentrated wave packets are just the linear combinations of the extended ones. 

Hence, the hypothesis about the new micro-particle physics, below some exceptional limits, can hardly be defended without modifying everything. In particular, the theories of non-commutative geometry in which the space coordinates fulfil $[x,y]=\sigma\neq0$ as a fundamental identity, could not be constructed. It is enough to note that then the simultaneous squeezing transformation $x\rightarrow\nu x$ and $y\rightarrow\nu y$ where $\nu\neq0$ would ruin the non-commutative law.

As essential consequences would follow from the inverse operations in which the initial wave packet could be amplified. Some time ago, a group of authors studied the properties of the radiation emitted from the {\em extended sources}, asking whether some properties of such sources can be reconstructed from the emitted radiation \cite{rzazewski92}. However, each extended  packet is a linear combination of the localized ones. The question then arises, what would happen if some experiments could {\em fish} in the emitted photon state some components corresponding to the localized parts of the initial source? Would the initial state  be reduced to one of its localized emission points as in the {\em delayed choice experiment} of Wheeler? The idea seemed unreal, so the authors of \cite{rzazewski92} worried rather about the momenta than space localization of the extended source. (The situation, however, could be different in case of  the amplified  particle states on a plane orthogonal to the symmetry axis of some solenoid or ion trap.) 

In fact, the possible state amplification in two dimensions, (e.g. around the  solenoid  axis),  would lead to unsolved  reduction problems no less challenging, such as {\em interaction free of measurement} by Elitzur and Vaidman \cite{elitzur93}.  Here, the finally measured position ${\bf \tilde{q}} = \lambda {\bf q}$ would be an observable commuting with the initial one.  Hence, the observer performing an (imperfect) position measurement in the future could obtain {\em a posteriori} the more exact data about its position in the past. The problem is, whether it does repeat the scenario of the {\em non-demolition measurement} \cite{nondemolishing78,nondemolishing80}. If so, does the reduction of the wave packet affects also the particle state in the past?  

However, the localization must cost some energy ({\em cf.} Wigner, Yanase {\em et al} \cite{yanase64,jauch67,yanase71}). One might suppose that the energy needed to localize the amplified packet in the future was provided by the reserves in the screen grains (or whatever medium in which the particle was finally absorbed). Yet, this would require an assumption that the relatively low energy invested  in the future should be pumped into some higher energy needed for  more precise localization  in the past. This  seems impossible --not just due to the causality paradox, but also, due to the energy deficit!

Some fundamental aspects of quantum mechanics seem to be in opposition. An undeserved return to the old discussions? Although,  one should not forget that all our techniques were based on exploring the evolution matrices \eqref{eq:2} which obey a strictly linear, orthodox quantum mechanics, albeit, is this theory indeed true? Whatever the answer is, it looks like the low energy phenomena might be as close (if not closer) to the fundamental problems of quantum theory as well as high energy physics.

%
%
%
%
\vspace{0.5cm}
\textsc{Acknowledgements}. The authors are indebted to their colleagues at the Department of Physics in CINVESTAV, Mexico City, the summer school in Bialowie$\dot{\text{z}}$a, Poland and the University of Oslo for their interest and helpful remarks. B. Mielnik would like to thank the support from the CONACYT project 152574.
\newpage
%
%
%
%
\appendix
\section{Picking the Suit Amplitudes for Squeezing}
\label{s:appendix}

In principle, the parameters $b,c$ and the function $\gamma$ in \eqref{eq:22} are arbitrary. However, all of these must be carefully chosen so as to construct $\beta$--amplitudes suitable to generate the squeezing (or amplification) operations, {\em i.e.} those positive valued functions $\beta$ vanishing at the beginning and at the end of the symmetric interval $[-\tau,\tau]$. 

We have found that another approach to pick up the appropriate amplitudes can be done by simply observe the polar map produced by the two eigenvalues of the matrix $u$ along the full interval, namely: {\em the squeezing effect will be achieved only if the pair of eigenvalues of $u$ do not generate any {\em complementary pedal curve} at any point in $[-\tau,\tau]$}.

For instance, consider the dotted amplitude on figure \ref{fig:2}(a). To this one corresponds the parameters $b=\frac{9}{5}$ and $c = \frac{7}{2}$, whereas $\gamma = \sin^{2}\tau$ and the initial condition $\beta_{0}$ has been fixed to zero. The matrix $u$ related to this $\beta$--amplitude possesses a pair of eigenvalues which real parts are depicted along the interval $[-\frac{\pi}{2},\frac{\pi}{2}]$ on figure \ref{fig:a}(a). One can notice the formation of kind of {\em complementary pedal curves}, hence, such amplitude would not produce squeezing or amplification. On the other hand, now consider the dashed amplitude on figure \ref{fig:2}(a). This pulse has been parametrized through $b=\frac{7}{4}$ and $c = -3$, with $\gamma = \sin^{2}\tau$ and $\beta_{0}=0$. The corresponding {\em eigentrajectories} of $u$, are depicted on figure \ref{fig:a}(b), as one can see, there are not {\em complementary pedal curves} at any point along the interval, concluding that this amplitude is suitable to produce the {\em soft operations}.

\begin{figure}[h!]
\begin{center}
$\vcenter{\hbox{\includegraphics[width=6.25cm]{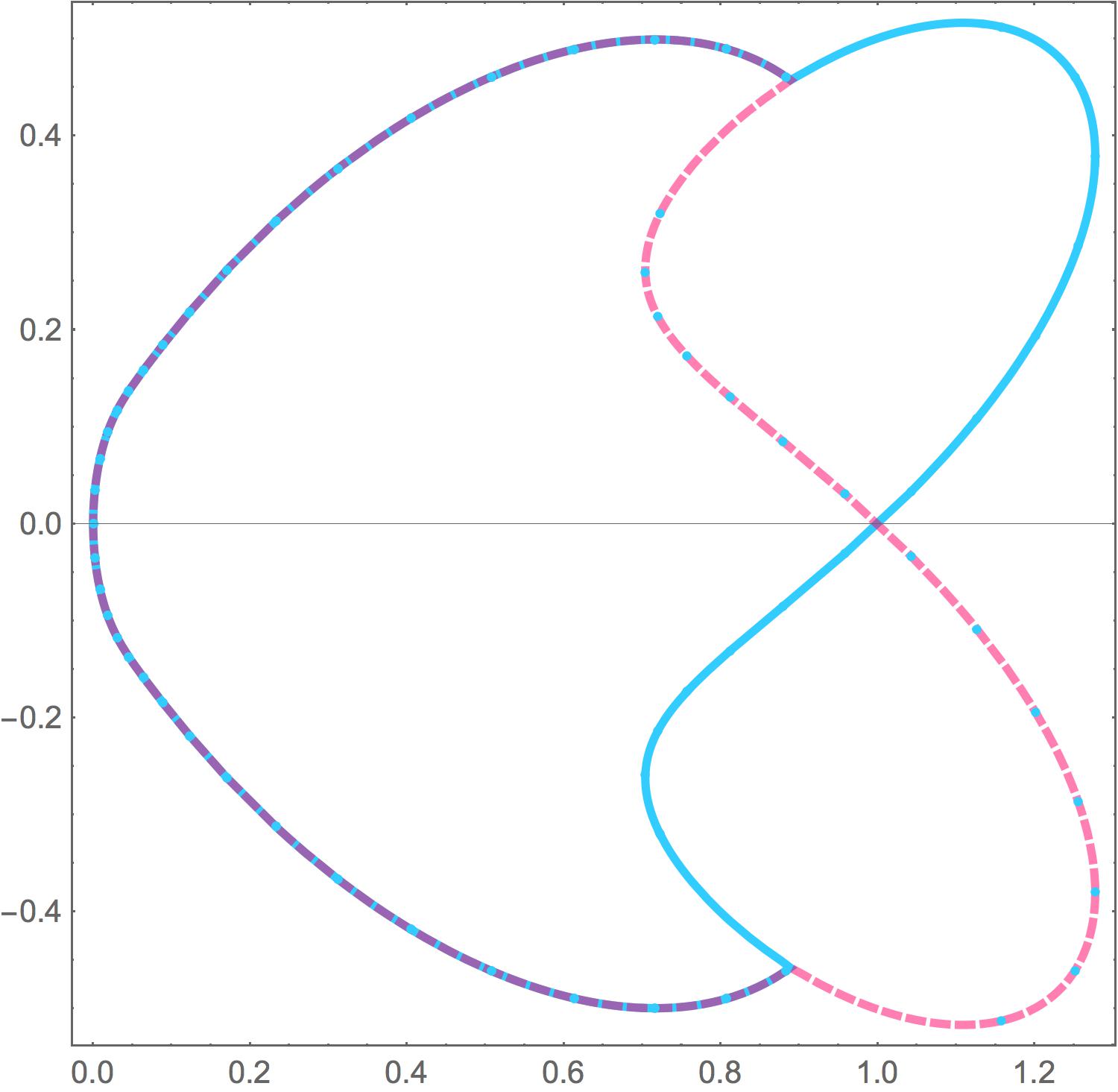}}}$
$\vcenter{\hbox{\includegraphics[width=6.25cm]{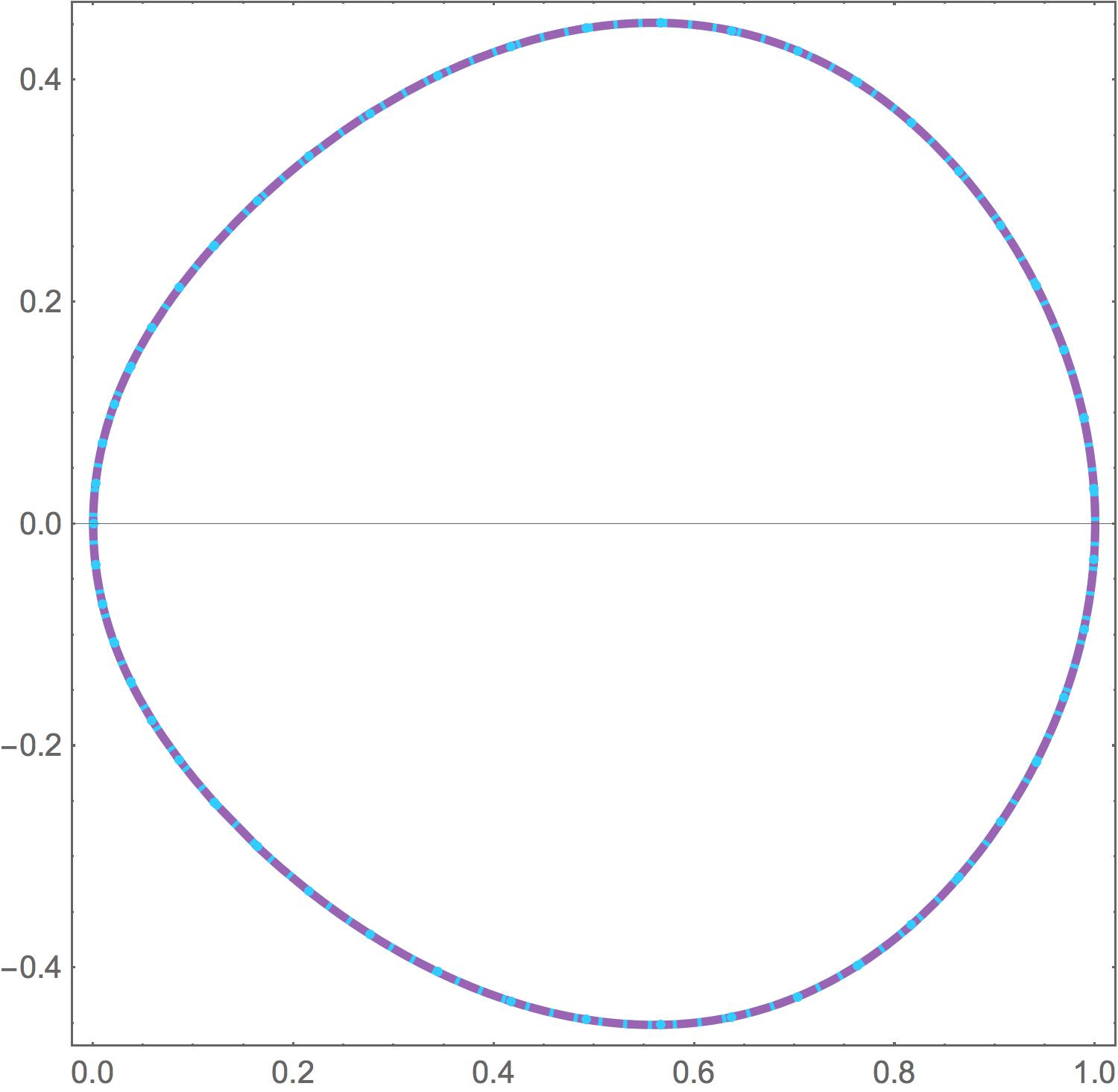}}}$
\\
$\vcenter{\hbox{\footnotesize (a) \hspace{5.5cm} (b)}}$
\caption{Eigentrajectories of the matrix $u$ for the dotted and dashed pulses on figure \ref{fig:2}(a), respectively.} 
\label{fig:a}
\end{center}
\end{figure}

%
%
%
%
\bibliographystyle{unsrt}
\bibliography{references}
\end{document}